\documentclass[12pt,a4paper]{article}

\usepackage[linesnumbered,vlined]{algorithm2e}
\usepackage{url}
\usepackage{tikz} 
\usepackage{amsmath,amssymb,amsfonts}
\usepackage{amsthm,amscd,bm}
\usepackage{authblk}
\usepackage[a4paper,top=3cm,bottom=3cm,left=3cm,right=3cm]{geometry}
\usepackage{makecell}

\newcommand{\cyl}[1]{{Cyl(#1)}}

\newcommand{\ZZ}{\mathbb{Z}}
\newcommand{\zzd}{{\mathbb{Z}^d}}

\newcommand{\alf}{\mathbb{S}}
\newcommand{\gr}{\mathbb{G}}
\newcommand{\grh}{\mathbb{H}}
\newcommand{\grs}{\mathbb{S}}
\newcommand{\grx}{\mathbb{X}}

\newcommand{\Ker}{\mathrm{Ker}}

\newcommand{\Imma}{\mathrm{Im}}
\newcommand{\locrule}{\ensuremath{f}}
\newcommand{\glorule}{\ensuremath{F}}

\usepackage{xspace, comment, mathdots}

\newcommand{\ie}{i.e.\@\xspace}

\DeclareMathOperator{\lcm}{lcm}

\newtheorem{theorem}{Theorem}
\newtheorem{definition}{Definition}
\newtheorem{remark}{Remark}

\newtheorem{lemma}{Lemma}

\title{Topological transitivity of group cellular automata is decidable }

\author[1]{Niccol\`o Castronuovo}
\author[2]{Alberto Dennunzio}
\author[3]{Luciano Margara}

\affil[1]{Liceo ``A. Einstein,''  Rimini, Italy}
\affil[2]{Department of Informatics, Systems and Communication, University of Milano-Bicocca, Italy}
\affil[3]{Department of Computer Science and Engineering, University of Bologna, Campus of Cesena, Italy}

\date{}
   
\begin{document}
\maketitle

\begin{abstract}
Topological transitivity is a fundamental notion in topological dynamics and is widely regarded as a basic indicator of global dynamical complexity. 
For general cellular automata,  topological transitivity is known to be undecidable. 
By contrast, positive decidability results have been established for one-dimensional group cellular automata over abelian groups, while the extension to higher dimensions and to non-abelian groups has remained an open problem.
In this work, we settle this problem by proving that topological transitivity is decidable for the class of $d$-dimensional ($d\geq 1$) group cellular automata over arbitrary finite groups. 
Our approach combines a decomposition technique for group cellular automata, reducing the problem to the analysis of simpler components, with an extension of several results from the existing literature in the one-dimensional setting.
As a consequence of our results, and exploiting known equivalences among dynamical properties for  group cellular automata, we also obtain the decidability of several related notions, including total transitivity, topological mixing and weak mixing, weak and strong ergodic mixing, and ergodicity.

\end{abstract}




\section{Introduction} \label{introduction}

A discrete-time dynamical system is a pair $( \grx,\glorule )$, where $\grx$ is a  topological space (the phase space) and $\glorule \colon  \grx \to  \grx$ is a continuous map. 
The system $( \grx,\glorule )$ is topologically transitive if for every pair of nonempty open sets $U,V \subseteq  \grx$ there exists an integer $n \ge 0$ such that $\glorule^n(U) \cap V \neq \varnothing$.
Topological transitivity is a fundamental concept in topological dynamics, reflecting the ability of the system to move between arbitrary regions of the phase space $ \grx$ under iterations of $\glorule$. As such, it has been extensively studied and is widely regarded as a basic indicator of global dynamical complexity, often serving as a stepping stone toward stronger properties such as mixing (see, for example, \cite{Auslander1988,Devaney1989,DurandPerrin2022,Glasner2003} and the references therein).

\medskip\noindent
Cellular automata (CAs) are discrete-time dynamical systems whose phase space is $\alf^\zzd$, where $\alf$ is a finite set of states, and whose global evolution map $\glorule:\alf^\zzd\to\alf^\zzd$ is a continuous function commuting with all the shift maps of $\alf^\zzd$. Here, $d$ is a positive integer representing the dimension of the phase space on which the CA is defined.
Group cellular automata (GCAs) are CAs whose phase space is $\gr^\zzd$ for a finite group $\gr$ and the map $\glorule \colon \gr^\zzd \to \gr^\zzd$ is a continuous group endomorphism commuting with all shifts of the space. 

\medskip\noindent
The global map $\glorule$ of any CA is determined by a local rule $f$ (Curtis-Hedlund-Lyndon-theorem \cite{hedlund69}). The local rule is a finite object and admits several equivalent representations, for instance in tabular form. This finitary description leads to the following natural question: given the description of the local rule $f$ defining $\glorule$, is it possible to decide algorithmically whether the system $(\alf^\zzd,\glorule )$ satisfies a prescribed dynamical property, such as topological transitivity?


\medskip\noindent
For general CAs, even  when restricting to reversible one-dimensional CAs, topological transitivity has been shown to be undecidable \cite{Lukkarila10}.  By contrast, for one-dimensional GCAs over abelian groups, topological transitivity is decidable \cite{DennunzioFGM21INS} and, moreover, admits an effective characterization \cite{DBLP:journals/isci/DennunzioFM24}. 


\medskip\noindent
In this paper we address the 
problem of determining whether topological transitivity is algorithmically decidable for GCAs beyond the one-dimensional abelian setting (see Question 2 in \cite{BeaurK24} and Question 1 in \cite{CASTRONUOVO2026103749}).
Our main contribution is a complete positive answer: we prove that topological transitivity is decidable for  GCAs in any dimension and over arbitrary finite groups (Theorem~\ref{supermain}).
The proof combines two main ingredients.

First, we develop a transitivity-preserving decomposition technique obtained by iteratively taking quotients by verbal subgroups.
Given a GCA $(\gr^\zzd,\glorule)$, the procedure ${\tt VerbalDecomposition}$ produces a finite family of GCAs
$\{(\gr_1,\glorule_1),\dots,(\gr_k,\glorule_k)\}$ whose underlying groups admit no non-trivial proper verbal subgroups.
We show that this decomposition preserves topological transitivity: $(\gr^\zzd,\glorule)$ is topologically transitive if and only if each $(\gr_i^{\zzd},\glorule_i)$ is topologically transitive (Theorem~\ref{th_transitivity}).
By a classical structural fact (Remark~\ref{rem:verbaliso}), each $\gr_i$ is a direct product of isomorphic simple groups, hence the original decision problem reduces to this highly constrained class.

Second, we establish decidability of topological transitivity on both branches of the classification.
For the abelian case, where $\gr\simeq(\ZZ/p\ZZ)^n$, we exploit the algebraic representation of GCAs by matrices over multivariate Laurent polynomials and extend to arbitrary dimension an effective transitivity criterion based on the characteristic polynomial (Theorem~\ref{dec-zpn-trans}, building on Theorem~\ref{grinberg}).
For the non-abelian case, where $\gr\simeq \grs^n$ with $\grs$ a finite non-abelian simple group, we extend the one-dimensional analysis of \cite{CASTRONUOVO2026103749} and obtain a decision procedure in the surjective setting (Theorem~\ref{dec-zpn-nonabelian}).

Putting these ingredients together yields an explicit algorithm deciding topological transitivity for any finite-group GCA (Section~5), thereby proving Theorem~\ref{supermain}.
Moreover, since for GCAs topological transitivity is equivalent to total transitivity, topological mixing, weak and strong ergodic mixing, and ergodicity (Section~\ref{sec: Mixing properties}), our results also imply decidability of all these properties.




\medskip\noindent
The rest of this paper is organized as follows.
In Section~\ref{CAandGCA} we recall the basic definitions and preliminary results on CAs and GCAs.
In Section~\ref{sec: Mixing properties} we review several dynamical properties and the relations among them, with a particular emphasis on the simplifications that occur in the GCA setting.
Section~\ref{decomposition_verbal} introduces our transitivity-preserving decomposition technique based on verbal subgroups and establishes the key reduction theorem.
In Sections~\ref{abisom} and~\ref{nonab} we prove decidability of topological transitivity for GCAs defined over direct products of isomorphic simple groups, treating the abelian and non-abelian cases, respectively.
Finally, in Section~\ref{conclusions} we conclude and discuss directions for further research.


\section{CAs and GCAs} \label{CAandGCA}

In this section, we review the fundamental definitions and basic results related to CAs and GCAs. For additional definitions and results, we refer the reader to those introduced in~\cite{GCA24}.

\medskip\noindent
Let $\alf$ be a finite set of states and let $d$ be any positive integer. 
A $d$-dimensional configuration over  $\alf$ is a function from $\zzd$ to $\alf$, \ie, an assignment of symbols
of $\alf$ on the infinite grid $\zzd$.   
$\alf^\zzd$ denotes the set of all the
$d$-dimensional configurations over $\alf$.
Given any configuration $c \in \alf^\zzd$ and any element $v\in \zzd$, the value of $c$ at position $v$ is denoted by $c_v$.
We equip $\alf^\ZZ$ with the prodiscrete topology, that is, the product
topology obtained by endowing each factor $\alf$ with the discrete topology.
With this topology, $\alf^\zzd$ is a compact topological space. 

\medskip\noindent
For an element $u\in \zzd$, the $u$-shift map  $\sigma_u: \alf^\zzd\to \alf^\zzd$ is defined as follows:
 $$\forall c\in \alf^\zzd,\ \forall v\in \zzd:\  \sigma_u(c)_v=c_{v+u}.
 $$
 Notice that $(\alf^\zzd,\sigma_u)$ itself is a CA.
\medskip\noindent
 A $d$-dimensional \emph{CA} on $\alf$ is any continuous function $\glorule: \alf^\zzd \to \alf^\zzd$ which is also shift commuting, \ie, $\glorule\circ \sigma_u=\sigma_u \circ \glorule$ for every $u\in \zzd$.

\medskip\noindent
Any $d$-dimensional CA  
can be equivalently defined \cite{hedlund69} by means of a  local rule  $f:\alf^k \to \alf$ together with a neighbor vector
$v_f=(v_1,\dots,v_k)$ of elements of $\zzd$ as follows:
$$\forall c \in \alf^\zzd,\ \forall v\in \zzd:\ \glorule(c)_v=f(c_{v+v_1},\dots, c_{v+v_k}).$$ 

\medskip\noindent
Let $\gr$ be a finite group with identity element $e$. The set $\gr^\zzd$ is
itself a group under the componentwise operation induced by the group operation
of $\gr$. We denote by $e^\zzd \in \gr^\zzd$  the
identity element of the group $\gr^\zzd$, i.e.,  the configuration taking
the value $e$ at every position $v\in\zzd$. Clearly, when equipped with the prodiscrete topology, $\gr^\zzd$ turns out to be a topological compact group. A configuration $c\in \gr^\zzd$ is said to be \emph{finite} if the number of positions $v\in\zzd$ such that $c_v\neq e$ is finite.

\medskip\noindent
A  CA $\glorule: \gr^\zzd\to \gr^\zzd$ is said to be a   group cellular automata (GCA) if $\glorule$ is an endomorphism of $\gr^{\zzd}$. In that case, the local rule of $\glorule$ is a homomorphism $f:\gr^k \to \gr$.  
Any group homomorphism
$
f : \gr^k \to \gr
$
is determined by a family of homomorphisms
$
h_i : \gr \to \gr 
$
whose images commute pairwise \cite{GCA24}, so that
\[
f(u_1,\dots,u_k) = \prod_{1\leq i\leq k} h_i\big(u_i)
\]
is well-defined.
We will write $f=(h_1,\dots,h_k)$. 

\medskip\noindent 
\textbf{Notation.}\\
- When the configuration space is clear from the context, we will refer to a cellular automaton simply by its local rule $f$ or by its global rule $\glorule$. 
In other situations, when it is convenient to specify the configuration space $ \grx$, we will denote a GCA by the pair $( \grx,\glorule)$.\\
- The term GCA will refer to a $d$-dimensional group cellular automaton, where $d$ is an arbitrary positive integer.  Whenever we specifically consider the one-dimensional case, this will be stated explicitly. \\
- We will say that a GCA is defined on the group $\gr$ to mean that its configuration space is $G^\zzd$.

\section{Dynamical properties and their relations } \label{sec: Mixing properties}

The system $( \grx, \glorule )$ is topologically transitive if for every pair of nonempty open sets $U,V \subseteq  \grx$ there exists an integer $n \ge 0$ such that $ \glorule ^n(U) \cap V \neq \varnothing$.


\medskip\noindent
Many other dynamical properties have been defined and studied in the literature 
for discrete-time dynamical systems. 
We briefly list some of them below, together with their definitions, which are given in a concise form since these properties will not be the focus of the present paper.

\medskip\noindent
-- $( \grx, \glorule )$ is \emph{totally transitive} if $( \grx, \glorule ^n)$ is topologically transitive for every $n\ge 1$.

\medskip\noindent    
-- $( \grx, \glorule )$ is \emph{topologically mixing} if for every pair of nonempty open sets $U,V\subseteq  \grx$ there exists $n_0\ge 1$ such that
$ \glorule ^n(U)\cap V \neq \varnothing$ for all $n\ge n_0.$

\medskip\noindent    
-- $( \grx, \glorule )$ is \emph{topologically weakly
mixing} if $( \grx\times \grx, \glorule \times \glorule)$ is transitive.

\medskip\noindent 
Let $( \grx,\mathcal{B},\mu)$ be a probability space and let $ \glorule : \grx\to  \grx$ be a measurable map preserving $\mu$.

\medskip\noindent  
-- $( \grx,\mathcal{B},\mu, \glorule )$ is \emph{ergodically strongly mixing}   if for all $A,B\in\mathcal{B}$ one has
\[
\mu\bigl( \glorule ^{-n}A\cap B\bigr)\xrightarrow[n\to\infty]{}\mu(A)\mu(B).
\]

\medskip\noindent 
-- $( \grx,\mathcal{B},\mu, \glorule )$ is \emph{ergodically weakly  mixing} if for all $A,B\in\mathcal{B}$ one has
\[
\frac{1}{N}\sum_{n=0}^{N-1}\left|\mu\bigl( \glorule ^{-n}A\cap B\bigr)-\mu(A)\mu(B)\right|
\xrightarrow[N\to\infty]{}0.
\]

\medskip\noindent   
-- $( \grx,\mathcal{B},\mu, \glorule )$ is \emph{ergodic} if 
\[
 \glorule ^{-1}(A)=A \ \text{(mod $\mu$)} \ \Longrightarrow \ \mu(A)\in\{0,1\}\quad  \text{  for all } A\in\mathcal{B}
\]

\medskip\noindent
For general dynamical systems we know that
\[
\text{topologically mixing}
\;\Longrightarrow\; \text{topologically weakly mixing}
\;\Longrightarrow\;
\text{totally transitive}\]
\[
\;\Longrightarrow\;
\text{topologically transitive},
\]
and
\[
\text{ergodically strongly mixing}
\;\Longrightarrow\;
\text{ergodically weakly mixing}
\;\Longrightarrow\;
\text{ergodic}.
\]

\medskip\noindent
In the setting of GCAs, the overall picture is considerably simpler.
In fact, it is not hard to show that, for  GCAs, topological transitivity, total transitivity, topological mixing, weak and strong ergodic mixing, and ergodicity are equivalent properties. In particular, for the case $d=1$ this equivalence is proved in \cite[Theorem~3]{CASTRONUOVO2026103749}, and the same argument extends with  little effort to the case $d>1$.

%
%
\section{A decomposition technique for  GCAs preserving 
topolological transitivity} \label{decomposition_verbal}

Let $\gr$ be a finite group, and let $\grh \trianglelefteq \gr$ be a fully invariant subgroup of $\gr$.
A fully invariant subgroup of a group $\gr$ is a subgroup $\grh\leq \gr$ such that, for every endomorphism $\phi$ of $\gr$, one has $\phi(\grh)\leq \grh$. 
For $g \in \gr$, we denote by $[g] := g\,\grh$ the coset of $g$ in the quotient group $\gr/\grh$.
Moreover, for $c \in \gr^\zzd$ we write $[c]$ for the element of $(\gr/\grh)^\zzd$ defined by
$[c]_v := [c_v]$ for every $v \in \zzd$.

Let $\glorule$ be a  GCA over $\gr$. Since $\grh$ is fully invariant then  $\glorule(\grh^\zzd)\subseteq \grh^\zzd$. 
The maps 
$$\overline{\glorule}:\grh^\zzd\to \grh^\zzd \ \text{ and }\ \widetilde{\glorule}:(\gr/\grh)^\zzd\to (\gr/\grh)^\zzd$$
are defined as follows:
\begin{eqnarray}
&&\forall c\in \grh^\zzd:\ \overline{\glorule}(c):=\glorule(c) \text{ and }\label{Fbar}\\
&&\forall [c]\in (\gr/\grh)^\zzd:\ \widetilde{\glorule}([c]):=[\glorule(c)]. \label{Ftilde}
\end{eqnarray}
Note that, by Equations~\eqref{Fbar} and~\eqref{Ftilde} and since $\glorule(\grh^\zzd)\subseteq \grh^\zzd$,  $(\grh,\overline{\glorule})$ and $((\gr/\grh)^\zzd,\widetilde{\glorule})$   are well-defined  GCAs.

We now recall the definition of a verbal subgroup. 
Let $w(x_1,\dots,x_n)$ be a group word in $n$ variables, i.e., a (reduced) word of the form
\[
w(x_1,\dots,x_n)=x_{i_1}^{\varepsilon_1}\cdots x_{i_k}^{\varepsilon_k},
\qquad i_j\in\{1,\dots,n\},\ \varepsilon_j\in\{\pm1\}.
\]
The verbal subgroup of $\gr$ associated with $w$ is the subgroup
\[
w(\gr)=\langle\, w(g_1,\dots,g_n)\mid g_1,\dots,g_n\in \gr\,\rangle,
\]
generated by all values of $w$ obtained by substituting arbitrary elements of $\gr$ into the variables $x_1,\dots,x_n$.
As an example, consider the word
\[
w(x_1,x_2)=x_1^{-1}x_2^{-1}x_1x_2.
\]
The associated verbal subgroup is the commutator subgroup of $\gr$.

It is well-known that a verbal subgroup is fully-invariant. 

\begin{remark}\label{rem:verbaliso}
 The only finite groups with no non-trivial proper verbal subgroups are direct products of isomorphic simple groups
 (see for example Theorem 8 in \cite{CASTRONUOVO2026103749}). 
\end{remark}
We now define a decomposition function for  
GCAs via quotienting by verbal subgroups.

\vspace{0.3cm}
  
\begin{algorithm}[H]

\SetAlgoNlRelativeSize{0}
\SetAlgoNoLine
\SetKwFunction{Decomposition}{VerbalDecomposition}
\SetKwProg{Fn}{Function}{:}{}
\Fn{\Decomposition{$\gr$,$\glorule$}}
{
\If{$\gr$ admits no non-trivial proper verbal subgroups}
{\KwRet{$\{(\gr,\glorule)\}$\;}}
\Else{
Let $\grh$ be a non-trivial proper verbal subgroup of $\gr$\;
a  $\gets$ \Decomposition{$\gr/\grh,\widetilde{\glorule}$}\;
b  $\gets$ \Decomposition{$\grh,\overline{\glorule}$}\;
\KwRet{$a\cup b$}\;}
}
\end{algorithm}

\vspace{0.3cm}

The function {\tt VerbalDecomposition} takes as input a GCA $(\gr,\glorule)$ and produces 
a finite collection $\{(\gr_1,\glorule_1),\dots,(\gr_k,\glorule_k)\}$, 
where each $\gr_i$ has no non-trivial proper verbal subgroups and  $\glorule_i$ is a GCA over $\gr_i^\zzd$ for all $i\in\{1,\dots,k\}$.
Notice that, by Remark~\ref{rem:verbaliso}, each $\gr_i$ is a direct product of isomorphic simple groups.


\medskip\noindent
We now have all the ingredients needed to state the following theorem.

\begin{theorem}\label{th_transitivity}
Let $\gr$ be a finite group and let $\glorule$ be a  GCA on $\gr^\zzd.$
Let $$\{(\gr_1,\glorule_1),\dots,(\gr_k,\glorule_k)\}$$ be the output produced by the function {\tt VerbalDecomposition} called on  $(\gr,\glorule)$.
Then $(\gr^\zzd,\glorule)$ is topologically transitive 
if and only if each $(\gr_i^\zzd,\glorule_i)$ is topologically transitive.
\end{theorem}
\begin{proof}
To prove this theorem it is sufficient to prove that for every verbal subgroup $\grh$ of $\gr$,
$(\gr^\zzd,\glorule)$ is topologically transitive if and only if both $(\grh^\zzd,\overline{\glorule})$ and 
$((\gr/\grh)^\zzd,\widetilde{\glorule})$ are topologically transitive. 
We will prove the following three statements separately.\\
\noindent
$(a)$  If $(\grh^\zzd,\overline{\glorule})$ and 
$((\gr/\grh)^\zzd,\widetilde{\glorule})$ are topologically transitive then $(\gr,\glorule)$ is topologically transitive.\\
\noindent
$(b)$  If $(\gr,\glorule)$ is topologically transitive then $((\gr/\grh)^\zzd,\widetilde{\glorule})$ is topologically transitive.\\
\noindent
$(c)$  If $(\gr,\glorule)$ is topologically transitive then $(\grh^\zzd,\overline{\glorule})$ is topologically transitive.

\medskip\noindent
The proof of statement $(a)$ is a straightforward adaptation of the proof of the corresponding statement for $d=1$ given in 
\cite[Theorem~4]{CASTRONUOVO2026103749}.
Notice that in this proof the authors make use of a result by Moothathu (Corollary~9.3 in \cite{Moothathu2005}), which states that the product of two topologically transitive CAs is topologically transitive. 
Although Moothathu’s result is formulated in the one-dimensional setting, the same conclusion holds for  GCAs over  $\gr^\zzd$. 
Indeed, a topologically transitive  GCA is always topologically weakly mixing (see \cite{Moothathu2009}), and the product of a topologically transitive  GCA with a topologically weakly mixing GCA is again topologically transitive. 
This follows from the results in \cite{Huang}, which are the same tools used by Moothathu to establish his Corollary~9.3.

\medskip\noindent
The proof of statement $(b)$ is identical to that of the corresponding result for $\gr^\ZZ$ given in \cite[Theorem~4]{CASTRONUOVO2026103749}.

\medskip\noindent
The proof of statement $(c)$ is technically the most demanding and relies on structural properties of verbal subgroups. 
It constitutes the crucial missing step in \cite{CASTRONUOVO2026103749} for having a complete proof of decidability of topological transitivity for  GCAs.

For the sake of readability, we prove statement $(c)$ only for a specific verbal subgroup, namely the commutator subgroup. This is the verbal subgroup $w(\gr)$ defined by the word $w(x,y)=xyx^{-1}y^{-1}$. The proof of the general case follows the same argument, replacing $w(x,y)$ with the word that defines the relevant verbal subgroup.

Let $k$ be any positive integer,  $M\in (\zzd)^k$, and
$P\in \gr^k$.
A cylinder $\cyl{M,P}$ of $G^{\mathbb{Z}^d}$ is 
defined by
$$
\cyl{M,P}=\left\{c\in \gr^\zzd:\ c_{M(i)}=P(i) \text{ for every } i\in \{1,\dots , k\} \right\}
$$
where $M(i)$ and $P(i)$ denote the $i$-th element of $M$ and $P$, respectively.

Cylinders  are clopen sets and  form a basis for the prodiscrete topology on $\gr^\zzd$.   Then
$(\gr^\zzd,\glorule)$ is topologically transitive if and only if for every   $k>0$,   every $M\in (\zzd)^k$, and  every pair $P,Q \in \gr^k$,   there exists $n\ge 0$ such that 
\begin{equation}\label{transcyl}
   \glorule^{n}\left(Cyl(M,P)\right)\cap Cyl(M,Q)\neq\varnothing. 
\end{equation}
We now prove that if Property \eqref{transcyl} holds, 
then  
$$\glorule^{n}\left(Cyl(M,P)\cap \grh^\zzd \right)\cap \left(Cyl(M,Q)\cap \grh^\zzd\right)\neq\varnothing,$$
or, equvalently, that the topological transitivity of $(\gr^\zzd,\glorule)$ implies the topological transitivity of $(\grh^\zzd,\overline{\glorule})$.

Let $k>0$, $M\in (\zzd)^k$ and 
$P,Q \in \grh^k$.
By definition of $\grh$,  
$S=\{ g_1g_2g_1^{-1}g_2^{-1}:\ g_1,g_2 \in \grh \}$
is a finite generating set for $\grh$.
Every element of $\grh$ can be written as a product of finitely many elements of $S$.
Since $\grh$ is finite, there exists a constant $q$, depending only on $\grh$, such that every element of $\grh$ can be expressed as a product of exactly $q$ elements of $S$ (this can be achieved by padding the product with a suitable number of elements equal to $e$).
As a consequence, in what follows, we will assume without loss of generality that the patterns $P$ and $Q$ introduced below take values in $S$ instead of in $\grh$.
\begin{eqnarray*}
 P&=&\left(x_1y_1x_1^{-1}y_1^{-1},\dots, x_ky_kx_k^{-1}y_k^{-1}\right) \\
 Q&=&\left(r_1s_1r_1^{-1}s_1^{-1},\dots, r_ks_kr_k^{-1}s_k^{-1}\right),
\end{eqnarray*}
where $x_i,y_i,r_i,s_i \in \grh$.
Let
\begin{eqnarray*}
 P_x=(x_1,\dots, x_k),&&  P_x^{-1} =(x_1^{-1},\dots, x_k^{-1}), \\
 P_y= (y_1,\dots, y_k),&&   P_y^{-1}= (y_1^{-1},\dots, y_k^{-1}),\\
Q_r= (r_1,\dots, r_k),&&  Q_r^{-1}= (r_1^{-1},\dots, r_k^{-1}),\\
Q_s= (s_1,\dots, s_k),&&  Q_s^{-1}= (s_1^{-1},\dots, s_k^{-1}).
\end{eqnarray*}
Using component-wise group multiplication, we have 
\begin{eqnarray*}
 P=P_xP_yP_x^{-1}P_y^{-1} &\text {and }&
 Q=Q_rQ_sQ_r^{-1}Q_s^{-1}.
\end{eqnarray*}

Since $(\gr^\zzd,\glorule)$ is topological transitive, 
we can find 
 $c_x \in \cyl{M,P_x}$, $c_y \in \cyl{M,P_y}$, $c_r \in \cyl{M,Q_r}$, and
$c_s \in \cyl{M,Q_s}$ such that 
\begin{eqnarray*}
 \glorule^{n}(c_x)=c_r &\text{and}&
  \glorule^{n}(c_y)=c_s
\end{eqnarray*}
and then 
\begin{eqnarray*}
 \glorule^n(c_xc_yc_x^{-1}c_y^{-1})&=&c_rc_sc_r^{-1}c_s^{-1}.
\end{eqnarray*}
Notice that, for a topologically transitive GCA $\glorule$, once the size of the cylinders is fixed, there exists a positive integer $n$ such that any pair of cylinders of that size can reach one another after exactly $n$ iterations of $\glorule$.
Also notice that 
$c_x,c_y,c_x^{-1},c_y^{-1},$ $c_r,c_s,c_r^{-1}, c_s^{-1}\in \gr^\zzd$,
but  
\begin{eqnarray*}
c_xc_yc_x^{-1}c_y^{-1}\in  \cyl{M,P}\cap \grh^\zzd &\text{and}& c_rc_sc_r^{-1}c_s^{-1}\in  \cyl{M,Q}\cap\grh^\zzd.
\end{eqnarray*}
Since $P$ and $Q$  are arbitrary elements of $\grh^k$,
it follows that
$(\grh^\zzd,\overline{\glorule})$ is topologically transitive.

\end{proof}

\section{ Topological transitivity of GCAs over direct products of isomorphic abelian simple groups}\label{abisom}

Direct products of isomorphic abelian simple groups are precisely the finite groups isomorphic to $(\ZZ/p\ZZ)^n$, where $p$ is a prime number and $n$ is a positive integer.
GCAs over $(\ZZ/p\ZZ)^n$ admit a natural representation in terms of multivariate Laurent polynomials and series.
Any configuration $c \in ((\ZZ/p\ZZ)^n)^\zzd$ is represented by a vector of length $n$ with entries in  
$(\ZZ/p\ZZ)[[x_1^{\pm 1},\dots,x_d^{\pm 1}]]$
and the map $\glorule$ is represented by an $n\times n$  matrix
with entries in $ (\ZZ/p\ZZ)[x_1^{\pm 1},\dots,x_d^{\pm 1}]$.
With this representation, the $n$-th iterate $\glorule^n(c)$ can be equivalently computed as $M^n v$, where $M$ is the matrix representation of $\glorule$ and $v$ is the vector representation of $c$.

In the study of dynamical properties of GCAs over abelian groups, the characteristic polynomial $\chi(t)$ of the matrix associated with $\glorule$ plays a fundamental role.
In particular, with respect to topological transitivity, we have the following results.

\begin{theorem}[\cite{DennunzioFGM2020INS}]\label{grinberg}
Let $\gr=(\ZZ/p\ZZ)^n$.
$( \gr^\ZZ,\glorule)$ is topologically transitive if and only if $\glorule$ is surjective and 
\begin{equation*}
    \gcd(\chi(t), t^{p^i-1}-1)=1\; \text{ for all } i\in\{1,\dots,n\},
\end{equation*}
where $\chi(t)$ is the characteristic polynomial 
of the matrix representing $\glorule$.
\end{theorem}

The proof of this theorem is purely algebraic and can be extended in a fairly natural way to the $d$-dimensional case, yielding the following result.

\begin{theorem}\label{dec-zpn-trans}
 Topological transitivity for  GCAs over $(\ZZ/p\ZZ)^n$ is decidable. 
\end{theorem}
\begin{proof}
The statement follows from an extension of Theorem~\ref{grinberg} to the $d$-dimensional setting, together with the fact that surjectivity is decidable for GCAs over $(\ZZ/p\ZZ)^n$ (see~\cite{kari2000}).
\end{proof}

%
\section{ Topological transitivity of GCAs over direct products of isomorphic non-abelian simple groups}\label{nonab}


Let $\grs$ be a finite non-abelian simple group and let
$\gr = \grs_1 \times \cdots \times \grs_m$,
where each $\grs_i$ is isomorphic to $\grs$. For every $i$, we identify $\grs_i$ with the canonical subgroup of $\gr$ whose elements have trivial components except possibly at the $i$-th coordinate. 

Let $\glorule$ be a surjective GCA over $\gr^\zzd$, and let $\locrule=(h_1,\dots,h_k)$ be its surjective local rule where each $h_i$ 
is an endomorphism of $\gr$.
Then the following statements hold (see \cite{CASTRONUOVO2026103749} for their proofs).

\medskip\noindent
$(a)$ There exists a partition $\{J_1,\dots,J_k\}$ of the set $\{1,\dots,m\}$ such that, for each $i\in\{1,\dots,k\}$,
$$\Imma(h_i)=\prod_{t\in J_i} \grs_t.$$

\medskip\noindent
$(b)$ There exist subsets $I_1,\dots,I_k \subseteq \{1,\dots,m\}$ such that, for each $i\in\{1,\dots,k\}$,
$$\Ker(h_i)=\prod_{t\in I_i} \grs_t.$$

\medskip\noindent
$(c)$ For every $t\in\{1,\dots,m\}$ there exists a unique index $i$ such that $\grs_t$ is not a factor of $\Ker(h_i)$.
 
\medskip\noindent
$(d)$ The restriction of every $h_i$ to a factor $\grs_j$ over which its action is nontrivial produces an automorphism $h'_i$ between $\grs_j$ and another component $\grs_{l}.$ 

\begin{definition}\label{pifo}
Let us define $\pi_{\locrule}$ as
the permutation of the symbols $\{1,2,...,m\}$ defined by 
$\pi_f(j)=l$ if and only if there exists an $i$ such that $h_i|_{\grs_j}=\grs_l.$ Denote by $o$ the order of the permutation $\pi_f$ as an element of the symmetric group over $\{1,2,...,m\}$.
\end{definition}

If $h_i$ acts non-trivially over $\grs_j,$ the map $\hat h_i$ obtained composing $h_i$ $o$ times with itself is an automorphism of $\grs_j.$ Denote by $o_i$ the order of $\hat h_i$ as an element of the group $Aut(\grs_j).$

\begin{definition}
A group $\gr$ which is the product $\grs_1\times\cdots \times \grs_m$ of finite, non-abelian, isomorphic simple groups $\grs_i$ is said to be \emph{minimal} with respect to the action of a given GCA $\glorule$ over $\gr$ if there are no partitions of $\{1,2,...,m\}$ into two sets $I,J$ such that $$\glorule\left(\left(\prod_{i\in I}\grs_i\right)^\zzd\right)\subseteq \left(\prod_{i\in I}\grs_i\right)^\zzd 
\ \ \text{and} \ \ 
\glorule\left(\left(\prod_{i\in J}\grs_i\right)^\zzd\right)\subseteq \left(\prod_{i\in J}\grs_i\right)^\zzd.$$
\end{definition}

When $\glorule$ is surjective, the group $\gr$ is minimal with respect to the action of $\glorule$
if and only if the corresponding permutation $\pi_f$ above defined is a single cycle.
\begin{remark}\label{minimality}
    If $\gr$ is not minimal with respect to $\glorule$ it is possible to decompose the dynamics of $\glorule$ into the product of the dynamics of $\glorule$ restricted to its minimal components. In particular a given GCA over $\gr$ is topologically transitive if and only if all its minimal components are topologically transitive. 
\end{remark}

\begin{lemma}\label{lemma_product_simples}
Let $G=\grs_1\times\cdots\times \grs_m$ be a product of finite, non-abelian, isomorphic simple groups. Let $\glorule$ be a surjective and minimal  
GCA over $\gr$ with  local rule $\locrule=(h_1,\dots,h_k)$ and neighbor vector
$v=(v_1,\dots,v_k)\in (\zzd)^k$.  
Let $\{J_1,\dots,J_k\}$ be the index partition introduced in statement~$(a)$ and let $r_i=|J_i|$ for all $i\in \{1,\dots,k \}$.
Let $o$ and $o_i$, $i\in \{1,\dots,m \}$, be defined as above. Then 
$$\glorule^{o\alpha}=\sigma_{-\alpha\beta}.$$  
whre $\alpha = \lcm(o_1,\dots,o_m)$ and $\beta=\sum_{i=1}^k r_i v_i.$
\end{lemma}
\begin{proof}
   The details of this proof in the one-dimensional case can be found in \cite{CASTRONUOVO2026103749}.
The extension to the $d$-dimensional case presents no additional difficulties and follows the same line of reasoning as in the one-dimensional setting.
\end{proof}

\begin{theorem}\label{dec-zpn-nonabelian}
 Topological transitivity for  GCAs over $\grs^n,$ where $\grs$ is a simple non-abelian group, is decidable. 
\end{theorem}
\proof
First of all it is decidable if a given  GCA over any finite group  is surjective \cite{BeaurK24}.
Hence we can decide if a  GCA over $\grs^n$ is surjective. Non-surjectivity immediately implies the lack of topological transitivity.
As a consequence, we can assume that $\glorule$ is surjective. By Remark \ref{minimality} we can also assume that $\glorule$ is minimal. 
Given the local rule of a GCA $\glorule$ over $\grs^n$, the quantities $o,$ $o_i$ and $r_i$ of Lemma \ref{lemma_product_simples} are  computable.   
Hence, to decide whether $\glorule$ is topologically transitive, it suffices to compute the vector $\alpha\beta$.
The GCA $\glorule$ is topologically transitive if and only if this vector is nonzero. 
\endproof

\section{Decidability of topological transitivity}
By combining the results of Sections~\ref{sec: Mixing properties},~\ref{decomposition_verbal}, and~\ref{nonab}, we can finally provide an algorithm that decides whether a  GCA over an arbitrary finite group is topologically transitive or not.

\vspace{0.3cm}
\begin{algorithm}[H]
\SetAlgoNlRelativeSize{0}
\SetAlgoNoLine
\SetKwFunction{test}{IsTransitive}
\SetKwProg{Fn}{Function}{:}{}
\Fn{\test{$\gr$,$\glorule$}}
{
\If{$(\gr,\glorule)$ is not surjective}
{\KwRet{False\;}}
$\{(\gr_1,\glorule_1),\dots,(\gr_k,\glorule_k)\}$  $\gets$ \Decomposition{$\gr$,\glorule}\;

\For{$i\gets 1$ \KwTo $k$}{
  \If{$(\gr_i^\zzd,\glorule_i)$ is not topologically transitive}{
    \KwRet{False\;}
  }
}
 \KwRet{True}\;
 }
\end{algorithm}
\vspace{0.3cm}

We now explain why the function \test is effectively computable and why its output determines whether $(\gr^\zzd,\glorule)$ is topologically transitive or not.

\medskip\noindent
\textbf{Computability of \test.}\\
Instruction~2 is computable: proved in \cite[Theorem 21]{BeaurK24}.\\
Instruction~5 is computable: since $\gr$ is a finite group, the computation of its verbal subgroups and the corresponding quotient operations are effective 
procedures.\\
Instruction~7 is computable: since each $(\gr_i^\zzd,\glorule_i)$ is a  GCA over a direct product of isomorphic simple groups, for which topological transitivity is decidable by Theorems~\ref{dec-zpn-trans} and~\ref{dec-zpn-nonabelian}.

\medskip\noindent
\textbf{Correctness of \test.}\\
The correctness of \test\ mainly follows from Theorem~\ref{th_transitivity}, which shows that $(\gr^\zzd,\glorule)$ is topologically transitive if and only if each $(\gr_i^\zzd,\glorule_i)$ is topologically transitive. 
Notice that, at Instruction~2, we perform a preliminary check on the surjectivity of $(\gr^\zzd,\glorule)$. 
This is required because the decidability of topological transitivity for GCAs defined over direct products of isomorphic non-abelian simple groups (Theorem~\ref{dec-zpn-nonabelian}) is established only for surjective GCAs. 
However, this restriction does not cause any loss of generality, since topological transitivity always implies surjectivity.

\medskip\noindent
Summarizing the results of this work, we can state the following theorem.

\begin{theorem}\label{supermain}
 Topological transitivity for  GCAs over a finite group $\gr$ is decidable. 
\end{theorem}


\section{Conclusions and Further Work}\label{conclusions}
In this paper we prove that topological transitivity is decidable for  group cellular automata over arbitrary finite groups, thus resolving the open questions raised in \cite{BeaurK24,CASTRONUOVO2026103749}. 
Together with the equivalence, in the class of  GCAs, between topological transitivity and several stronger (or a priori different) dynamical properties, our main theorem yields the decidability of total transitivity, topological mixing, weak and strong ergodic mixing, and ergodicity.

The present work suggests a number of natural directions for further research.

\medskip\noindent
\textbf{Group subshifts.}
A first direction is to move beyond full shifts and consider \emph{group subshifts}, that is, closed shift-invariant subgroups $ \grx\subseteq \gr^{\mathbb{Z}^d}$, endowed with the induced prodiscrete topology. 
In this setting one can study continuous shift-commuting endomorphisms $\glorule: \grx\to  \grx$ (and, more generally, $\gr$-equivariant endomorphisms for suitable actions), which provide a natural analogue of GCAs on constrained configuration spaces.
A key problem is whether topological transitivity (and the equivalent mixing/ergodic properties) remains decidable when the phase space is a group subshift given by a finite description (e.g., of finite type, sofic, or specified by a finite set of forbidden patterns).
Identifying classes of group subshifts where a decomposition approach still applies, or where transitivity can be reduced to an algebraic condition on finitely presented modules, appears particularly promising.

\medskip\noindent
\textbf{Amenable-group index sets.}
A second direction concerns the extension from the lattice $\mathbb{Z}^d$ to more general indexing groups.
Given a countable amenable group $\grh$, one can consider the full shift $\gr^{\grh}$ with the (pro)discrete product topology and the shift action of $\grh$ by left translations.
It is then natural to define \emph{group cellular automata over $\grh$} as continuous $\grh$-equivariant group endomorphisms $\glorule:\gr^{\grh}\to \gr^{\grh}$, and to ask whether topological transitivity is decidable in this broader setting.
Here amenability provides a robust replacement for Følner geometry, which is often crucial in the study of mixing properties, entropy, and recurrence; nevertheless, the lack of a linear order and the potentially complex geometry of $\grh$ may require new tools.
    A concrete goal is to determine which parts of the decomposition technique extend from $\zzd$ to amenable groups, and to isolate subclasses of $\grh$ (e.g., virtually abelian, polycyclic, or more generally groups with an effective Følner sequence) for which an algorithmic characterization can be obtained.

\medskip\noindent
\textbf{Monoid-based configuration spaces.}
A further natural direction is to investigate whether the results obtained in this paper can be extended from groups to more general algebraic structures, in particular to monoids.
One may consider cellular automata whose configuration space is $M^{\mathbb{Z}^d}$, where $M$ is a finite monoid, and whose global evolution map is a continuous shift-commuting monoid endomorphism.
While many of the tools used in the group setting rely heavily on invertibility and on the availability of normal and verbal subgroups, it is natural to ask to what extent analogous decomposition techniques can be developed for monoids, possibly based on congruences or Green’s relations.
A central question is whether topological transitivity (and related mixing or ergodic properties) remains decidable in this broader setting, and which algebraic features of the underlying monoid play a decisive role.
Addressing these questions would contribute to a deeper understanding of the interplay between algebraic structure and dynamical complexity in cellular automata.

\medskip\noindent
\textbf{Complexity and effective invariants.}
Finally, beyond decidability it is natural to investigate the computational complexity of the decision procedures arising from our approach, and to develop effective invariants (e.g., module-theoretic or representation-theoretic data) that control transitivity and mixing for GCAs.
This may lead to sharper classifications and to algorithms that are practical for concrete instances.

\medskip\noindent
We believe that these directions provide a fruitful framework for extending the algorithmic theory of dynamical properties from full group shifts over $\mathbb{Z}^d$ to more general algebraic and geometric settings.

\bibliographystyle{plain}
\bibliography{bib}

\end{document}